\documentclass{LMCS}

\def\dOi{9(4:18)2013}
\lmcsheading%
{\dOi}
{1--17}
{}
{}
{Apr.~20, 2013}
{Dec.~\phantom05, 2013}
{}

\ACMCCS{[{\bf Theory of computation}]: Logic}

\usepackage{hyperref}

\usepackage{amsmath}
\usepackage{amssymb}

\renewcommand{\emptyset}{\varnothing}
\newcommand{\N}{\mathbb{N}}

\def \eqcard{{EqCard}} 
\newcommand{\I}{\mathcal{I}}
\newcommand{\C}{\mathcal{C}}

\def \cm{{\mathcal M}}

\begin{document}

\title[Expansions of MSO by cardinality relations]{Expansions of MSO by cardinality relations}

\author[Alexis B\`es]{Alexis B\`es}	
\address{Universit\'e Paris-Est, LACL (EA 4219),  UPEC, F-94010, Cr\'eteil, France}	
\email{bes@u-pec.fr}  







\begin{abstract}
  \noindent We study expansions of the Weak Monadic Second Order theory of $(\N,<)$ by {\em cardinality relations}, which are predicates $R(X_1,\dots,X_n)$ whose truth value depends only on the cardinality of the sets $X_1,\dots,X_n$.

We first provide a (definable) criterion for definability of a cardinality relation in $(\N,<)$,\break and use it to prove that for every cardinality relation $R$ which is not definable in $(\N,<)$,\break there exists a {\em unary} cardinality relation that is definable in $(\N,<,R)$ and not in $(\N,<)$. These results resemble Muchnik and Michaux-Villemaire theorems for Presburger Arithmetic.  We prove then that $+$ and $\times$ are definable in $(\N,<,R)$ for every cardinality relation $R$ which is not definable in $(\N,<)$.  This implies undecidability of the WMSO theory of $(\N,<,R)$. 

We also consider the related satisfiability problem for the class of finite orderings, namely the question whether an MSO sentence in the language $\{<,R\}$ admits a finite model $M$ where $<$ is interpreted as a linear ordering, and $R$ as the restriction of some (fixed) cardinality relation to the domain of $M$. We prove that this problem is undecidable for every cardinality relation $R$ which is not definable in $(\N,<)$.

\end{abstract}

\maketitle

\section*{Introduction}

The B\"uchi-Elgot-Trakhtenbrot Theorem \cite{Buchi60,Elgot61,Trakh61} states the equivalence between rational languages and relations definable in the Weak Monadic Second Order theory of the structure $(\N,<)$ (for short: $\mbox{WMSO}(\N,<)$), and yields decidability of this theory. As an easy consequence, the satisfiability problem for Monadic Second-Order (for short: MSO) logic over the class of finite orderings is also decidable. These results initiated the development of many logical formalisms for the specification and automatic verification of systems (see e.g. \cite{henriksen1995mona}). 

A natural issue is to add expressiveness to $\mbox{WMSO}(\N,<)$ while keeping good properties such as robustness and decidability.
B\"uchi proved that $\mbox{MSO}(\N,<)$ is also decidable \cite{Buchi62}. Since then, there have been many works on (un)decidable extensions of $\mbox{WMSO}(\N,<)$ and $\mbox{MSO}(\N,<)$. Let us state some classical examples.  It is known that $\mbox{WMSO}(\N,+)$, and even $\mbox{WMSO}(\N,<,x \mapsto 2x)$, 
are undecidable \cite{Robinson58}. Elgot and Rabin study in \cite{ElgotRabin66} the MSO theory of structures of the form $(\N,<,P)$, where $P$ is some unary predicate. They give a sufficient condition on $P$ which ensures decidability of the MSO theory of $(\N,<,P)$. In particular the condition holds when $P$ denotes the set of factorials, or the set of powers of any fixed integer. The frontier between decidability and undecidability of related theories was explored in numerous later papers (see for instance \cite{Rabinovich07a}). 

Another approach is to extend $\mbox{WMSO}(\N,<)$ with predicates which express cardinality relations between finite sets. One can consider for instance the equi-cardinality relation  $|X|=|Y|$, which we denote by $EqCard(X,Y)$. On the one hand, Feferman and Vaught proved \cite[Theorem 8.2]{FV59} that $\mbox{WMSO}(\N, EqCard)$ (without $<$) is decidable, by reduction to Presburger Arithmetic; the interested reader can find applications of this result to constraint databases \cite{Rev05} and verification \cite{KNR05}.  
On the other hand it is easy to show that $+$ and $\times$ are definable in $\mbox{WMSO}(\N, <, EqCard)$, from which we get undecidability of the theory (see Proposition \ref{prop:equi-undec}). Klaedtke and Ruess  \cite{KlaedtkeR03} extend the undecidability result to the case where $EqCard$ is replaced by any predicate $R(X_1,\dots,X_r,Y_1,\dots,Y_s)$ which holds iff $|X_1|+\dots + |X_r| <|Y_1|+ \dots + |Y_s|$. They also introduce the notion of {\em Parikh automaton}, which allows them to obtain decidability results for some syntactical fragment of the logic.

The above results suggest that one cannot add very expressive cardinality relations to $\mbox{WMSO}(\N,<)$ while keeping decidability of the theory. Bruno Courcelle suggested to consider the case of predicates $R_A(X)$ which hold iff $|X|$ belongs to some (fixed) recursive set $A \subseteq \N$. The study is interesting only if $A$ is not ultimately periodic, since otherwise $R_A$ is already definable in $(\N,<)$. For which such sets $A$ can we obtain decidability of $\mbox{WMSO}(\N,<,R_A)$, and decidability of the related satisfiability problem ?  In \cite[Section 7.5]{Cour12}  it was shown that if $A$ is the union of powers of 2 and powers of 3 then the satisfiability problem for the corresponding logic over the class of finite orderings is undecidable. The proof goes by interpreting grids. It follows that $\mbox{WMSO}(\N,<,R_A)$ is undecidable. The general case was left open. 

In this paper we show that these undecidability results hold for every set $A \subseteq \N$ which is not ultimately periodic. We prove actually that they hold for any non-trivial expansion of $\mbox{WMSO}(\N,<)$ by a predicate $R(X_1,\dots,X_n)$ whose truth value depends only on the cardinality of the sets  $X_1,\dots,X_n$. We call such predicates {\em cardinality relations}. The predicates $R_A$, $EqCard$, as well as Klaedtke-Ruess' predicates which we discussed above, are examples of cardinality relations. 

In Section \ref{sec:mv} we provide a characterization of cardinality relations which are definable in $\mbox{WMSO}(\N,<)$.
Then 
we prove in Section \ref{sec:exp} that for every $n$-ary cardinality relation $R$ which is not definable in $\mbox{WMSO}(\N,<)$, there exists a unary cardinality relation $R'$ that is definable in $\mbox{WMSO}(\N,<,R)$ and not in $\mbox{WMSO}(\N,<)$. We also prove that $+$ and $\times$ are  definable in $\mbox{WMSO}(\N,<,R')$ for every unary cardinality relation $R'$ which is not definable in $\mbox{WMSO}(\N,<)$. As a corollary we obtain undecidability of  $\mbox{WMSO}(\N,<,R)$ for every cardinality relation which is not definable in $\mbox{WMSO}(\N,<)$. 

In Section \ref{sec:fini} we turn to the finite-model-theoretic version of the problem, i.e. we consider the same logical formalism but we interpret MSO formulas over structures of the form $(D,<^D,R^D)$ where $(D,<^D)$ is a finite initial segment of $(\N,<)$ and $R^D$ is the restriction of some fixed cardinality relation $R$ to $D$. We prove that for every $n$-ary cardinality relation $R$, if  $R$ is not definable in $\mbox{WMSO}(\N,<)$ then the satisfiability problem for MSO logic over the signature $\{<,R\}$ is undecidable. This solves in particular Courcelle's  question for predicates $R_A$. The proof essentially consists in defining arithmetic on (arbitrarily great) initial segments of the domain. We also study some particular cases of cardinality relations $R$ for which we can define $+$ and $\times$.   

The results and techniques of the present paper (mainly Sections \ref{sec:mv} and \ref{sec:exp}) are very close in spirit to a series of definability and decidability results \cite{Muc03,MV96,Bes97a} related to Presburger Arithmetic, i.e. the FO theory of $(\N,+)$. We will discuss these connections at the end of Sections \ref{sec:mv} and \ref{sec:exp}.

We also note that the study of logics which allow to express cardinality properties by means of special cardinality quantifiers is a classical topic in model theory (see  \cite{kaufmann1985quantifier,Cour12}). For recent results related to MSO logic, we refer e.g. to \cite{DBLP:journals/jsyml/BaranyKR11}, and also \cite{Colc13} which proves decidability of some extension of MSO.

We assume that the reader has some familiarity with MSO logic and automata theory. In particular, when proving that some property is definable by an MSO formula, we will often describe the construction of the formula but omit its full formal definition. 

\section{Preliminaries}

We denote by $\mathcal F$ the set of finite subsets of $\N$. Given $X \in {\mathcal F}$, $|X|$ will denote the cardinality of $X$.

For every $r \in \N$, we denote by $\N^n_{\geq r}$ the set of $n$-tuples $(x_1,\dots,x_n) \in \N^n$ such that  $x_i \geq r$ for every $i$, and define $\N^n_{<r}$ as  $\N^n \setminus \N^n_{\geq r}$.

We say that $X \subseteq \N$ is ultimately periodic if there exist integer $M$ and $p$ with $p \geq 1$ such that for every integer $x\geq M$ we have ($x \in X$ iff $x+p \in X$).

Given $j,k \in \N$, we denote by $[j,k]$ (respectively $(j,k)$) the interval $[\min{(j,k)},\max{(j,k)}]$ (resp. $(\min{(j,k)},\max{(j,k)})$ ). We also use notations such as $[j,k)$ and $(j,k]$ which are defined in a similar way.

Let $S$ be a set, $n \geq 1$, and let $R \subseteq {S}^n$ be an $n$-ary relation over $S$. Given $i \in [1,n]$ and $C \in S$, we denote by $R_{i,C}$ the $(n-1)$-ary relation over $S$ obtained from $R$ by fixing the $i$-th component to $C$, i.e. 
$$R_{i,C}=\{(x_1,\dots, x_{i-1},x_{i+1},\dots,x_n) \ | \ (x_1,\dots, x_{i-1},C,x_{i+1},\dots,x_n) \in R \}.$$

The relations $R_{i,C}$ will be called {\em sections} of $R$.

\subsection{Logic}

Let us briefly recall useful elements of monadic second-order logic. For more details see e.g. \cite{Gurevich85,Thomas97a}. Monadic second-order logic is an extension of first-order logic  that allows to quantify over elements as well as subsets of the domain of the structure. Given a signature ${\mathcal L}$, one can define  
monadic second-order formulas over ${{\mathcal L}}$ (which we call $\mathcal L$-formulas) as well-formed formulas that can use first-order variable symbols $x,y,\dots$ interpreted as elements of the domain of the structure, monadic second-order variable symbols $X,Y,\dots$ interpreted as subsets of the domain, symbols from ${{\mathcal L}}$, and a new binary  predicate $x \in X$ interpreted as membership relation. A sentence is a formula without free variable. 

When the context is clear, we often identify logical symbols with their interpretation. Otherwise the interpretation of a symbol $R$ in a structure $\mathcal M$ will be denoted by $R^{\mathcal M}$.

 Given a signature ${{\mathcal L}}$ and an  ${{\mathcal L}}$-structure ${\mathcal M}$ with domain $D$, we say that a relation $R \subseteq D^m \times (2^D)^n$ is {MSO-definable} in ${\mathcal M}$ if there exists a ${\mathcal L}$-formula $\varphi(x_1,\dots,x_m,X_1,\dots,X_n)$ which holds in ${\mathcal M}$ if and only if $(x_1,\dots,x_m,X_1,\dots,X_n)$ is interpreted by an $(m+n)$-tuple of $R$.  Given a structure ${\mathcal M}$ we denote by $\mbox{MSO}({\mathcal M})$ (respectively $\mbox{FO}({\mathcal M})$) the monadic second-order (respectively first-order) theory of ${\mathcal M}$.

{\em Weak} Monadic Second Order Logic (for short: WMSO) is obtained by restricting the interpretation of second order monadic variables to {\em finite} subsets of the domain of the structure. The notions of WMSO-definability and WMSO theory are defined similarly as above. In Section \ref{sec:fini} we will deal only with interpretations of formulas over finite structures; obviously in this context the distinction between WMSO and MSO is meaningless. 

From now on, definable will mean WMSO-definable.

Given two signatures ${{\mathcal L}}$ and ${{\mathcal L}'}$ such that ${{\mathcal L}} \subsetneq {{\mathcal L}}'$, an $\mathcal L$-structure $\mathcal M$ and an ${\mathcal L}'$-structure ${\mathcal M}'$ over the same domain, we say that ${\mathcal M}'$ is an expansion of ${\mathcal M}$ if every symbol of $\mathcal L$ has the same interpretation in ${\mathcal M}$ and ${\mathcal M}'$. Moreover we say that ${\mathcal M}'$ is a non-trivial expansion of ${\mathcal M}$ if there exists some symbol of ${{\mathcal L}'}$ whose interpretation in ${\mathcal M}'$ is not definable in ${\mathcal M}$. 

\subsection{Logic and finite automata}

In this section we recall classical results on WMSO logic and finite automata, and fix some notations. 
B\"uchi, Elgot and Trakhtenbrot independently proved that relations definable in $\mbox{WMSO}(\N,<)$ correspond to rational languages. This correspondance relies on the encoding of (tuples of) elements of $\mathcal F$. We define the function $c: {\mathcal F} \rightarrow \{0,1\}^*$ which maps every set $X \in {\mathcal F}$ to the finite word $c(X)$ over the alphabet $\{0,1\}$ defined by 
\begin{itemize}
\item $c(\emptyset)=\varepsilon$
\item if $X \ne \emptyset$ then $c(X)$ is a word of length $l=1+ \max(X)$ such that for every $i \in [0,l)$, the $i$-the letter of $c(X)$ equals $1$ iff $i \in X$.  
\end{itemize}
For instance if $X=\{0,3,4\}$ then $c(X)=10011$.

We also need to deal with $n$-tuples of elements of $\mathcal F$ for every fixed $n \geq 1$. We encode every $n$-tuple $\overline{X}=(X_1,\dots,X_n)\in {\mathcal F}^n$ by adding a sufficient number of zeros to the words $c(X_1)\dots,c(X_n)$ in such a way that they have the same length, and then ``stack up" these $n$ words, from which we obtain a word over the alphabet 
 $\Sigma_n=\{0,1\}^n$ of $n$-tuples of elements of $\{0,1\}$. Formally we extend the definition of $c$ to a function $c: {\mathcal F}^n \rightarrow \Sigma_n^*$ which maps every $n$-tuple $\overline{X}=(X_1,\dots,X_n)\in {\mathcal F}^n$ to the word $c(\overline{X}) \in \Sigma_n^*$ with length $l=1+max(X_1 \cup \dots \cup X_n)$ such that for every $i \in [0,l)$ and every $j \in [1,n]$, the $j$-th component of the $i$-th letter of $c(\overline{X})$ equals $1$ iff $i \in X_j$. 
 
For instance if $n=2$, $X_1=\{0,3,4\}$ and $X_2=\{0,1,3\}$ then $c(X_1)=10011$, $c(X_2)=1101$, and $c(X_1,X_2)={1 \choose 1} {0 \choose 1} {0 \choose 0} {1 \choose 1} {1 \choose 0}$.

\begin{thm}[B\"uchi-Elgot-Trakhtenbrot Theorem \cite{Buchi60,Elgot61,Trakh61}]\label{thm:buchi}
\ 
\begin{enumerate}
\item 
Let $n \geq 1$ and $R \subseteq {\mathcal F}^n$. The relation $R$ is definable in $(\N,<)$ iff the language $L_R \subseteq \Sigma_n^*$ defined by 
$$L_R=\{c(\overline{X}) \ | \ \overline{X} \in R \}$$
is rational.
\item $\mbox{WMSO}(\N,<)$ is decidable.

\end{enumerate}

\end{thm}

\begin{exas}\label{ex:card}

\begin{itemize}
\item The relation $R_1(X)$ which holds iff $X=[0,q)$ for some integer $q \geq 1$ is  definable in $(\N,<)$. Indeed  $L_{R_1}$ corresponds to the rational expression $11^*$. 
\item The relation $R_2(X)$ which holds iff $|X|$ is even, is definable in $(\N,<)$. Indeed  $L_{R_2}$ is the set of words over $\{0,1\}$ which have an even number of $1'$s and do not end with a $0$, which is a rational language. More generally it is easy to check that for every ultimately periodic set $A$, the relation $R_A(X)$ which holds iff $|X| \in A$ is definable in $(\N,<)$.

\item The relation $R_3(X_1,X_2)$ which holds iff $|X_1|$ and $|X_2|$ have the same parity, is also definable in $(\N,<)$. Indeed $L_{R_3}$ is the set of words $w$ over the alphabet $\Sigma_2$ which do not end with ${0 \choose 0}$, and such that the number of ${1 \choose 0}$'s and  the number of ${0 \choose 1}$'s in $w$ have the same parity. This is a rational language.

\item The relation $\eqcard(X_1,X_2)$ which holds iff $|X_1|=|X_2|$, is not definable in $(\N,<)$. The language $L_{\eqcard}$ is the set of words $w$ over $\Sigma_2$  which have the same number of  ${1 \choose 0}$ and ${0 \choose 1}$, and do not end with ${0 \choose 0}$. This language is not rational. 
\end{itemize}

\end{exas}

\noindent In the sequel we will consider definability in expansions of $(\N,<)$. The following auxiliary relations and functions will be useful. The proof is easy and left out.

\begin{prop}
The following functions and relations are definable in $(\N,<)$:
\begin{itemize} 
\item the union (resp. intersection) of finite sets;
\item the relation $Interval(x,y,X)$ which holds if $X=[x,y]$; 
\item the relation $Consec(x,y,Z)$ which holds if $x$ and $y$ are consecutive elements of $Z$;
\item For every $k \in \N$, the predicate $CARD_k(X)$ which holds iff $|X|=k$; 
\item For every $k \in \N$, the predicate $CARDLESS_k(X)$ which holds iff $|X| < k$. 
\end{itemize}

\end{prop}

\section{WMSO-definability of cardinality relations in $(\N,<)$}\label{sec:mv}

In this section we introduce {\em cardinality relations} and provide a characterization of cardinality relations which are definable in $(\N,<)$.

\begin{defi}
Let $n \geq 1$, and let $R \subseteq {\mathcal F}^n$. We say that $R$ is a {\em cardinality relation} if for 
all sets $X_1,\dots,X_n,Y_1,\dots,Y_n \in {\mathcal F}$,  if $|Y_i|=|X_i|$ for every $i \in [1,n]$, then we have 
$(X_1,\dots,X_n) \in R$ iff $(Y_1,\dots,Y_n) \in R$.
\end{defi}

For instance the relation $R_1$ of Example  \ref{ex:card} is not a cardinality relation, while $R_2$, $R_3$ and $\eqcard$ are.

We can associate to every $n$-ary cardinality relation $R$ the relation $\I(R)\subseteq {\N}^n$ defined by
$$\I(R)=\{(|X_1|,\dots,|X_n|) \ | \ (X_1,\dots,X_n) \in {R}\}.$$
Conversely, given  $S \subseteq {\N}^n$ we denote by $\C(S)$ the unique cardinality relation $R \subseteq {\mathcal F}^n$ such that $\I(R)=S$.

Given $n \geq 1$ and $S \subseteq \N^n$, we say that $S$ is {\em Card-definable} in a structure $\cm$ with domain $\N$ if $\C(S)$ is definable in $\cm$.
Given $k,l \geq 1$ and $f: \N^k \to \N^l$, we say that $f$ is {Card-definable} in $\cm$ if the graph of $f$ is {Card-definable} in $\cm$.

Our characterization of cardinality relations $R$ which are definable in $(\N,<)$ relies on periodicity conditions on $\I(R)$. We need the following definition.

\begin{defi}\label{def:strup}
Let $n \geq 1$, $S \subseteq \N^n$, and $\mu=(m,p_1,\dots,p_n)$ be an element of $\N^{n+1}$ such that $p_i \geq 1$ for every $i \in [1,n]$.

 We say that $S$ is {\em $\mu$-strongly ultimately periodic} (for short: $\mu$-STRUP)
 if for every $n$-tuple $(x_1,\dots,x_n) \in \N^n$ and every $j \in [1,n]$, if $x_i \geq m$ for every $i \in [1,n]$ then we have $(x_1,\dots,x_n) \in S$ iff $(x_1,\dots,x_{j-1},x_j+p_j,x_{j+1},\dots,x_n)  \in S$.

We say that $S$ is {\em strongly ultimately periodic}  (for short: STRUP) if there exists $\mu$ such that $S$ is $\mu$-STRUP.

\end{defi}

\begin{rem}
Observe that for $n=1$, $S$ is STRUP iff $S$ is ultimately periodic. 

\end{rem}

\begin{exas}\label{ex:strup}
\

\begin{itemize}
\item We consider the relations $R_2, R_3$ and $EqCard$ of Examples \ref{ex:card}. The relations $\I(R_2)=\{2x \ | \ x \geq 0\}$ and $\I(R_3)=\{(x,y) \ | \ x \equiv y \pmod 2\}$ are STRUP, while the relation $\I(EqCard)=\{(x,x) \ | \  x \geq 0\}$ is not STRUP.
\item Consider the relation $R_4(X,Y)$ which holds iff ($|X|$ is prime and $|Y| \leq 3$). Then the relation $\I(R_4)$ is $\mu$-STRUP with $\mu=(4,1,1)$. 

\end{itemize}
\end{exas}

\noindent The following theorem provides a characterization of cardinality relations which are definable in $(\N,<)$. 

\begin{thm}\label{thm:caract}
Let $n \geq 1$, and let $R \subseteq {\mathcal F}^n$  be a cardinality relation. The following properties are equivalent:

\begin{enumerate}
\item $R$ is definable in $(\N,<)$;
\item $\I(R)$ is a finite union of sets of the form $E_1 \times \dots \times E_n$, where each $E_i$ is an  ultimately periodic subset of $\N$ (i.e., $\I(R)$ is a recognizable subset of $\N^n$);
\item The two following properties hold:
\begin{enumerate}
\item every section of $\I(R)$ is Card-definable in $(\N,<)$;
\item $\I(R)$ is STRUP.

\end{enumerate}
\end{enumerate}

\end{thm}

\noindent Observe that item $(3)(a)$ could have been replaced by the property that every section of $R$ is definable in $(\N,<)$. This comes from the fact that $R$ being a cardinality relation, every section $[\I(R)]_{i,C}$ of $\I(R)$ equals $\I(R_{i,X})$ for any $X$ such that $|X|=C$.

\begin{exas}\label{ex:cardrel2222}

\begin{itemize} 
\item The relations $R_2$ and $R_3$ from Examples \ref{ex:card} satisfy items $3(a)$ and $3(b)$.
\item The relation $EqCard$ satisfies item $3(a)$ but not $3(b)$.
\item The relation $R_4$ from Example \ref{ex:strup} satisfies  $3(b)$, but not $3(a)$ since the section $[{\mathcal I}(R_4)]_{2,0}$ is the set of primes, which is not ultimately periodic. 
\end{itemize}

\end{exas}

\begin{proof}
The fact that $(2)$ implies $(1)$ is a straightforward consequence of the fact that every ultimately periodic subset $E$ of $\N$ is Card-definable in $(\N,<)$ (see Example \ref{ex:card}).

The proof that $(2)$ implies $(3)$ is easy and left out.

Let us prove that $(1)$ implies $(2)$. For every $(x_1,\dots,x_n) \in \N^n$, we have $(x_1,\dots,x_n) \in \I(R)$ iff the $n$-tuple

\centerline{$\overline{X}=([0,x_1),[x_1,x_1+x_2),\dots,[\sum_{1 \leq i <n} x_i, \sum_{1 \leq i \leq n} x_i))$}

belongs to $R$, which is in turn equivalent to   $c(\overline{X}) \in L_R$. Now we have $c(\overline{X})=a_1^{x_1} \dots a_n^{x_n}$ where for every $i \in [1,n]$, $a_i$ denotes the symbol $(b_1,\dots,b_n) \in  \Sigma_n$ such that $b_i=1$, and 
$b_j=0$ for every $j \ne i$. Hence we proved that $(x_1,\dots,x_n) \in \I(R)$ iff $a_1^{x_1} \dots a_n^{x_n}$ belongs to $L_R$, or equivalently to $L'_R=L_R \cap a_1^*\dots a_n^*$. Now $L_R$ is rational by Theorem \ref{thm:buchi}, thus $L'_R$ is a rational subset of $a_1^*\dots a_n^*$, and it is easy to check that every rational subset of $a_1^*\dots a_n^*$ is a finite union of sets of the form  $a_1^{E_1}\dots a_n^{E_n}$ where every $E_i$ is an ultimately periodic subset of $\N$. This yields $(2)$.

Finally we  prove that $(3)$ implies $(1)$. By $(3)(b)$, the set $\I(R)$ is $\mu$-STRUP for some $\mu=(m,p_1,\dots,p_n)$. 
 We have to prove that $R$ is definable in $(\N,<)$, i.e. that $\I(R)$ is Card-definable in $(\N,<)$. It suffices to prove that both sets $A^-=\I(R) \cap \N^n_{<m}$ and $A^+=\I(R) \cap \N^n_{ \geq m}$ are Card-definable in $(\N,<)$. 

By definition, we have $(x_1,\dots,x_n) \in A^-$ iff there exists $i \in [1,n]$ and $C<m$ such that $x_i=C$ and $(x_1,\dots,x_{i-1},x_{i+1},\dots,x_n) \in [\I(R)]_{i,C}$. By $(3)(a)$ each section $[\I(R)]_{i,C}$ is Card-definable in $(\N,<)$, thus the same holds for $A^-$.

Consider now the set $A^+$. Let
$$T=\{(j_1,\dots,j_n) \ | \ 0 \leq j_i <p_i \textnormal{ for every $i$, and }  (m+j_1,\dots,m+j_n) \in \I(R)  \}.$$
By definition of $\mu$ and $T$ we have
$$A^+=\bigcup_{(j_1,\dots,j_n) \in T} \{(m+j_1+k_1p_1,\dots,m+j_n+k_np_n) \ | \ k_1,\dots,k_n \geq 0 \}$$
and one can check that each set which appears in the union above is Card-definable in $(\N,<)$ (using again the fact that every ultimately periodic subset $E$ of $\N$ is Card-definable in $(\N,<)$), thus the same holds for $A^+$. 
\end{proof}

\

The above theorem provides a criterion for definability of a cardinality relation $R$ in $(\N,<)$. Actually this criterion is expressible with the signature $\{<,R\}$, i.e. we can find a sentence (independent of $R$) which holds in $(\N,<,R)$ iff $R$ is definable in $(\N,<)$. In order to state this property precisely, we need to distinguish (momentarily) between a relational symbol $R$ and its interpretation $R^{\mathcal M}$ in some structure $\mathcal M$.

\begin{prop}\label{prop:autodef}
Let $\{R_n(X_1,\dots,X_n) \ | \ n \geq 1\}$ denote a set of relational symbols. 
For every $n \geq 1$ 
there exists a $\{<,R_n\}$-sentence $\psi_n$   such that for every $\{<,R_n\}$-structure ${\mathcal M}=(\N,<,R_n^{\mathcal M}) $ where $R_n^{\mathcal M}$ is a cardinality relation, we have ${\mathcal M} \models \psi_n$ iff  $R_n^{\mathcal M}$ is definable in $(\N,<)$.

\end{prop}

\begin{proof}  We rely on item $(3)$ of Theorem \ref{thm:caract}.
For every $n \geq 1$ we define the sentence $\psi_n$ as $\psi_n^a \wedge \psi_n^b$ where $\psi_n^a$ and $\psi_n^b$ express conditions $3(a)$ and $3(b)$, respectively. The proof goes by induction over $n$.

We start with the sentence $\psi_n^b$, which can actually be constructed for each $n$ without relying on the induction hypothesis.

Indeed consider the $\{<,R_n\}$-formula 

$FSTRUP_n(Y,P_1,\dots,P_n)$:
$$ ( \bigwedge_{1 \leq i  \leq n} Y \cap P_i = \emptyset \wedge  \bigwedge_{1 \leq i  \leq n} P_i
 \ne \emptyset \wedge $$
 $$\forall X_1 \dots \forall X_n ((\bigwedge_{1 \leq i  \leq n} (X_i \cap P_i = \emptyset \wedge Y \subseteq X_i)) \rightarrow $$
$$ \ \ \bigwedge_{1 \leq j  \leq n} (R_n(X_1,\dots,X_n) \leftrightarrow R_n(X_1, \dots, X_{j-1}, X_j \cup P_j,  X_{j+1},\dots,X_n))).$$
This formula expresses, on the one hand, that all sets $P_i$ are nonempty and disjoint from $Y$, and on the other hand that $\I(R_n^{\mathcal M})$ is $\mu$-STRUP with $\mu=(|Y|,|P_1|,\dots,|P_n|)$.

Then it is clear that $\I(R_n^{\mathcal M})$ is STRUP (i.e. satisfies $(3)(b)$) iff the sentence 
$$\psi_n^b: \ \exists Y \exists P_1 \dots \exists P_n \ FSTRUP(Y,P_1,\dots,P_n)$$ holds in ${\mathcal M}$. 

We now turn to the construction of $\psi_n$ by induction over $n$. We only have to construct the sentence $\psi_n^a$. The case $n=1$ is straightforward since in this case property $(3)(a)$ is always true (a section of a unary relation is always Card-definable in $(\N,<)$). Now assume that $n \geq 2$ and the claim holds for every $n'<n$, and consider an $n$-ary relational symbol $R_n(X_1,\dots,X_n)$. Condition $(3)(a)$ can be expressed by saying that each section of $R_n^{\mathcal M}$ is definable in $(\N,<)$. This can be done by a slight modification of the sentence $\psi_{n-1}$, which existence is ensured by our induction hypothesis. More precisely we define $\psi_n^a$ as
$$\bigwedge_{1\leq i \leq n} \forall Z_i \ {\phi}_{i,{n-1}}(Z_i)$$
where for every $i$, the formula $\phi_{i,n-1}(Z_i)$ is constructed from $\psi_{n-1}$ by replacing in $\psi_{n-1}$ every occurrence of formulas of the form $R_{n-1}(T_1,\dots,T_{n-1})$ (where each $T_j$ denotes a monadic second-order variable) by $R_n(T_1,\dots,T_{i-1},Z_i,T_{i},\dots,T_{n-1})$. 
\end{proof}

\begin{rem}

Theorem \ref{thm:caract} and Proposition \ref{prop:autodef} can be seen as variations on Muchnik's results \cite{Muc03} on expansions of Presburger Arithmetic.
Indeed Theorem \ref{thm:caract} resembles Theorem 1 of \cite{Muc03}, which states that a relation $R \subseteq \N^n$ is FO definable in $(\N,+)$ iff all sections of $R$ are FO definable in $(\N,+)$, and moreover $R$ satisfies some additional periodicity conditions. And Proposition \ref{prop:autodef} is the counterpart of Theorem 2 of \cite{Muc03}, which states that the previous criterion for FO definability of a relation $R$ in $(\N,+)$ is expressible with the signature $\{+,R\}$.
\end{rem}

\section{Expansions of $\mbox{WMSO}(\N,<)$ by cardinality relations}\label{sec:exp}

In this section we prove that $+$ and $\times$ are definable in any (non-trivial) expansion of $(\N,<)$ by a cardinality relation, and obtain undecidability of the theory as a corollary.

A natural example of a cardinality relation is the equi-cardinality relation $\eqcard(X,Y)$. As noted in Examples \ref{ex:card}, this relation is not definable in $(\N,<)$.  The following result can be obtained as a straightforward consequence of results of \cite{Robinson58}, but we give a short self-contained proof (which will be useful in Section \ref{sec:fini}).

\begin{prop}\label{prop:equi-undec}
The functions $+$ and $\times$ are definable in $(\N,<,\eqcard)$. Therefore $\mbox{WMSO}(\N,<,\eqcard)$ is undecidable. 
\end{prop}

\begin{proof}
We first define $+$ using the fact that $z=x+y$ iff we have ($x \leq z$ and $|(x,z]|=|[0,y)|$), which we can express in $(\N,<,\eqcard)$. Then we can define $\times$ as follows: for all integers $x,y,z \geq 1$, we express that $z=xy$ by saying that the set of multiples of $y$ less than or equal to $z$ contains $z$ and has cardinality $x+1$. In other words, we have $z=xy$ if and only if there exists a finite set $A$ such that
\begin{itemize} 
\item $A \subseteq [0,z]$
\item $A$ contains both $0$ and $z$
\item $|A|=|[0,x]|$ 
\item for every pair $a_1,a_2$ of consecutive elements of $A$, we have $|[a_1,a_2)|=|[0,y)|$.
\end{itemize}
These properties can be defined easily in $(\N,<,\eqcard)$. Then undecidability follows from the undecidability of $\mbox{FO}(\N,+,\times)$.
\end{proof}

Now we consider the case of unary cardinality relations. By Theorem \ref{thm:caract}, we know that if $R$ is a unary cardinality relation, then $R$ is definable in $(\N,<)$ iff $\I(R)$ is ultimately periodic.

\begin{prop}\label{prop:unaire}
Let $R \subseteq {\mathcal F}$ be a unary cardinality relation that is not definable in $(\N,<)$.\break
The functions $+$ and $\times$ are definable in $(\N,<,R)$. Thus $\mbox{WMSO}(\N,<,R)$ is undecidable. 
\end{prop}

\begin{proof} 
By Proposition \ref{prop:equi-undec} it suffices to prove that the relation $\eqcard$ is definable in $(\N,<,R)$.
Given $x,y \in \N$, we claim that $x=y$ iff the following property, which we denote by $P(x,y)$, holds:
$$\forall z ( z+x \in \I(R) \leftrightarrow z+y \in \I(R)).$$
Obviously if $x=y$ then $P(x,y)$ holds. Now if $x \ne y$, say $x<y$ (without loss of generality), then $P(x,y)$ cannot hold, otherwise $\I(R)$ would be ultimately periodic with period $p=y-x$, i.e. would be STRUP, which by Theorem \ref{thm:caract}  contradicts our hypothesis that $R$ is not definable in $(\N,<)$.

Then we can define the relation $\eqcard$ in $(\N,<,R)$ by the formula $F(X,Y)$:
\begin{equation}\label{eq:F}
\forall Z ( Z \cap (X \cup Y)= \emptyset \rightarrow (R(X \cup Z) \leftrightarrow R(Y \cup Z))).
\end{equation}
Indeed it is easy to check that $F(X,Y)$ holds iff $P(|X|,|Y|)$ holds .
\end{proof}

The proof above shows that if $R$ is a unary cardinality relation which is not definable in $(\N,<)$, then one can define $EqCard$ in $(\N,<,R)$. Conversely the following holds.

\begin{prop}\label{prop:eqcardversunaire}
There exists a unary cardinality relation $R$ that is definable in $(\N,<,EqCard)$ and not definable in $(\N,<)$. 
\end{prop}

\begin{proof}
We can choose e.g. $R(X)$ as the relation which holds iff $|X|$ is a square. The set $\I(R)$ is not ultimately periodic thus by Theorem \ref{thm:caract}, $R$ is not definable in $(\N,<)$. Moreover $R(X)$ is definable in $(\N,<,EqCard)$ by a formula which expresses that either $X=\emptyset$, or there exists a set $Y$ such that $EqCard(X,Y)$ holds, and $Y=[0,x)$ where $x$ is a square. The latter property is definable in  $(\N,\times)$, hence by Proposition \ref{prop:equi-undec} it is definable in $(\N,<,EqCard)$.
\end{proof}

We consider now expansions of $(\N,<)$ by cardinality relations of any arity. Given $n \geq 1$ and any cardinality relation $R \subseteq {\mathcal F}^n$ which is not definable in $(\N,<)$, we shall prove that $+$ and $\times$ are definable in $(\N,<,R)$ by proving that there exists some {\em unary} cardinality relation $R'$ that is definable in $(\N,<,R)$ and not in $(\N,<)$, from which the result will follow by Proposition \ref{prop:unaire}.

\begin{thm}\label{thm:MV}
Let $n \geq 1$, and let $R \subseteq {\mathcal F}^n$ be a cardinality relation. If $R$ is not definable in $(\N,<)$ then there exists some unary cardinality relation $R' \subseteq {\mathcal F}$ that is definable in $(\N,<,R)$ and not definable in $(\N,<)$.

An equivalent formulation is: a cardinality relation $R$ is definable in $(\N,<)$ iff every unary cardinality relation definable in $(\N,<,R)$ is definable in $(\N,<)$.
\end{thm}

\begin{proof}
The proof proceeds by induction over $n$. The case $n=1$ is trivial. Assume now that $n \geq 1$ and that the claim is true for every $n' \leq n$. Let  $R \subseteq {\mathcal F}^{n+1}$ be a cardinality relation which is not definable in $(\N,<)$. By Theorem \ref{thm:caract}, either some section of $R$ is not definable in $(\N,<)$, or $\I(R)$ is not STRUP. 

If there exists some section $R_{i,C}$ of $R$ which is not definable in $(\N,<)$, then the result follows from the application of the induction hypothesis to the $n$-ary cardinality relation ${\tilde{R}}(X_1,\dots,X_{i-1},X_{i+1},\dots, X_{n+1})$:  
$$\exists X_i (CARD_C(X_i) \wedge R(X_1,\dots,X_{n+1}))$$ that is definable in $(\N,<,R)$ and not in $(\N,<)$.

Assume now that $\I(R)$ is not STRUP. Then we can define the relation $\eqcard$ in $(\N,<,R)$ by the formula $FF(X,Y)$ defined as:
{\setlength\arraycolsep{3pt}
\begin{eqnarray}\label{eq:FF}
\lefteqn{
\forall Z_1 \dots \forall Z_n (\bigwedge_{i=1}^n (Z_i \cap (X \cup Y)= \emptyset \rightarrow
} 
\\
& & (R(Z_1,\dots, Z_{i-1}, Z_i \cup X,Z_{i+1},\dots,Z_n) 
\leftrightarrow  R(Z_1,\dots, Z_{i-1},Z_i \cup Y,Z_{i+1},\dots,Z_n))). \nonumber
\end{eqnarray}
}%
This formula generalizes the formula $F(X,Y)$ defined in the proof of Proposition \ref{prop:unaire}.
It is clear that if $|X|=|Y|$ then $FF(X,Y)$ holds. Conversely if $FF(X,Y)$ holds for some $X,Y \in {\mathcal F}$ then $|X|=|Y|$, since otherwise  $\I(R)$ would be $\mu$-STRUP with
$$\mu=(\min(|X|,|Y|),||X|-|Y||,\dots,||X|-|Y||)$$
and this contradicts our hypothesis.
\end{proof}

\begin{rem}
With a careful examination of the above proof and using Proposition \ref{prop:autodef}, it is possible to extract an uniform definition of $R'$ from $R$, that is, for every $n \geq 1$ and every $n$-ary relational symbol $R(X_1,\dots,X_n)$ one can find a $\{<,R\}$-formula $S(X)$ such that if $R$ is not definable in $(\N,<)$ then the same holds for $S$.  
\end{rem}

We can state the main theorem of our section.

\begin{thm}\label{thm:indeci}
Let $n \geq 1$, and let $R \subseteq {\mathcal F}^n$ be a cardinality relation which is not definable in $(\N,<)$. The functions $+$ and $\times$ are definable in $(\N,<,R)$. Thus $\mbox{WMSO}(\N,<,R)$ is undecidable. 
\end{thm}

\begin{proof}
By Theorem \ref{thm:MV} there exists a unary cardinality relation that is definable in $(\N,<,R)$ but not in $(\N,<)$. The result follows from Proposition \ref{prop:unaire}.
\end{proof}

\begin{rem}\label{rem:FO}
Theorem \ref{thm:MV} resembles Michaux-Villemaire theorem for Presburger Arithmetic \cite{MV96}, which states that a relation $R \subseteq \N^k$ is FO definable in $(\N,+)$ iff every unary relation $R' \subseteq \N$ which is FO definable in $(\N,+,R)$ is FO definable in $(\N,+)$. 

Michaux-Villemaire result was used in \cite{Bes97a} to prove undecidability results for a class of expansions of Presburger Arithmetic. Theorem \ref{thm:indeci} can be seen as a variant of this work.

Our results deal with WMSO logic while \cite{Bes97a,MV96} consider FO logic. We can actually re-formulate our results in FO logic. Let us explain the main ideas.
Given $k \geq 2$, the so-called {\em B\"uchi-Bruy\`ere Arithmetic} (of base 
$2$) (see \cite{BHMV94}) is the structure $(\N,+,V_2)$ where $V_2(x)$ is the greatest power of $2$ which divides $x$ (with $V_2(0)=0$). There exists a strong connection between $\mbox{WMSO}(\N,<)$ and $\mbox{FO}(\N,+,V_2)$ (see \cite{villemaire1992joining} and \cite[Section 2.6]{bes02a}). Indeed consider the function $cod: {\mathcal F} \rightarrow \N$ which maps every $X \in {\mathcal F}$ to the integer $cod(X)=\sum_{i \in X} 2^i$. The function $cod$ extends naturally to a function $cod:{\mathcal F}^n \rightarrow \N^n$ for every $n \geq 1$. One can prove that a relation $R \subseteq {\mathcal F}^n$ is (WMSO) definable in $(\N,<)$ iff $cod(R)$ is FO-definable in $(\N,+,V_2)$. This allows to transfer (un)decidability results between theories $\mbox{WMSO}(\N,<,R)$ and theories $\mbox{FO}(\N,+,V_k,cod(R))$. In particular Theorem \ref{thm:indeci} yields the following result (which we state without proof).

\begin{prop} For every $n \geq 1$ and every cardinality relation $R \subseteq {\mathcal F}^n$ such that $cod(R)$ is not FO-definable in $(\N,+,V_2)$, the function $\times$ is FO-definable in $(\N,+,V_2,cod(R))$. It follows that $\mbox{FO}(\N,+,V_2,cod(R))$ is undecidable.
\end{prop}

For instance, the Proposition above applies to the structure
$(\N,+,V_2,EqNonZeroBits)$, where $EqNonZeroBits(x,y)$ holds iff the
binary representations of $x$ and $y$ have the same number of non-zero
bits; indeed we have $EqNonZeroBits=cod(EqCard)$.
\end{rem}

\section{The satisfiability problem for finite orderings}\label{sec:fini}

\subsection{Introduction}

In this section we consider the satisfiability problem for MSO logic in the signature $\{<,R\}$ over the class of finite orderings, where $R$ is interpreted as a cardinality relation. 

It should be noted that the question initially raised by Bruno Courcelle deals with satisfiability over the class of labelled graphs, which contains the class of finite orderings. Our undecidability results extend easily to the class of labelled graphs.

Let us fix some notations and definitions. From now on, we interpret formulas in finite structures, 
thus there is no need to distinguish between MSO and WMSO.
Given relations $R_1,\dots,R_k$ over ${\mathcal F}$ of any arity, we interpret $\{<,R_1,\dots,R_k\}$-formulas in structures of the form 
$${{\mathcal M}_q}=([0,q),<^{{\mathcal M}_q},R_1^{{\mathcal M}_q},\dots,R_k^{{\mathcal M}_q})$$ where $q \geq 1$, and $<^{{\mathcal M}_q},R_1^{{\mathcal M}_q},\dots,R_k^{{\mathcal M}_q}$ correspond to the restriction of $<,R_1,\dots,R_k$ (respectively) to the interval $[0,q)$. We shall simply write ${{\mathcal M}_q}=([0,q),<,R_1,\dots,R_k)$.

Given $q\geq 1$ and a $\{<,R_1,\dots,R_k\}$-formula $\varphi(X_1,\dots,X_n)$, we denote by $Val_q(\varphi)$ the relation defined by $\varphi$ in ${\mathcal M}_q$. If $Val_q(\varphi)$ is a cardinality relation (more precisely, the restriction of a cardinality relation to subsets of $[0,q)$), then we denote by ${\mathcal I}_q(\varphi)$ the relation ${\mathcal I}(Val_q(\varphi))$.

Besides the notion of definability of a relation in a fixed structure ${\mathcal M}_q$, we also consider a notion of definability over the class of structures ${\mathcal M}_q$:  given an $n$-ary relation $X$ over ${\mathcal F}$, we will say that $X$ is {\em $\{<,R_1,\dots,R_k\}$-finite-definable} if there exists a $\{<,R_1,\dots,R_k\}$-formula $\varphi(X_1,\dots,X_k)$ 
such that $Val_q(\varphi)=X \cap [0,q)^n$ for every $q \geq 1$. The previous definitions extend naturally to predicates with first-order free variables.

Given a signature $\mathcal L$, the {\em satisfiability problem for $MSO_{\mathcal L}$ over finite orderings} is the (decision) problem of whether a given  ${\mathcal L}$-sentence holds in at least one structure ${\mathcal M}_q$.

We first state some well-known (un)decidability results for the satisfiability problem. The following is an easy consequence of decidability of $\mbox{WMSO}(\N,<)$ (Theorem \ref{thm:buchi} item $(2)$) and the fact that the property ``$X$ is a proper initial segment of $\N$" is definable in $(\N,<)$.

\begin{prop}\label{prop:finitebuchi}
The satisfiability problem for $MSO_{\{<\}}$ over finite orderings is decidable.
\end{prop}

On the other hand, if we consider the signature $\{+,\times\}$ where the interpretation of $+$ (respectively $\times$) is the restriction of the graph of addition (resp. multiplication) to the domain of the structure, then we have the following.

\begin{prop}\label{prop:trakh}
\cite{Trakh50} \ 
The satisfiability problem for $MSO_{\{+,\times\}}$ over finite orderings is undecidable.
\end{prop}

\subsection{Cardinality relations and finite structures}

In the sequel we deal with signatures of the form $\{<,R\}$ where $R$ denotes a cardinality relation. A first natural question is whether Theorem \ref{thm:caract}, which provides a characterization of cardinality relations which are definable in $(\N,<)$, still hold in the context of finite orderings. The answer is positive.

\begin{prop}\label{prop:caract-fini}
For every integer $n \geq 1$ and every cardinality relation $R \subseteq {\mathcal F}^n$, the relation $R$ is $\{<\}$-finite-definable iff $R$ is definable in $(\N,<)$.
\end{prop}

\begin{proof}(sketch)
The "only if" direction is an easy consequence of the fact that ``$X$ is a proper initial segment of $(\N,<)$" is definable in $(\N,<)$. For the converse we can use item $(2)$ of Theorem \ref{thm:caract}, and the fact that the unary relation ``$|X| \in E$" is $\{<\}$-finite-definable for every ultimately periodic set $E$.  
\end{proof}

Now let us consider the satisfiability problem for $MSO_{\mathcal L}$ when ${\mathcal L}=\{<,R\}$ where $R$ is a cardinality relation. The following is a version of Proposition \ref{prop:equi-undec} for finite orderings.

\begin{prop}\label{prop:equi-undec-fini}
The graphs of $+$ and $\times$ are $\{<,\eqcard\}$-finite-definable. Therefore the satisfiability problem for $MSO_{\{<,\eqcard\}}$ over finite orderings is undecidable.
\end{prop}

\begin{proof}
The defining formulas given for $+$ and $\times$ in the proof of Proposition \ref{prop:equi-undec} still hold in the finite case. Undecidability follows from Proposition \ref{prop:trakh}.  
\end{proof}

Let us point out some difficulties which arise when one tries to adapt proofs of Proposition \ref{prop:unaire} and Theorem \ref{thm:MV} in the context of finite orderings. Consider first the case of unary cardinality relations. In the proof of Proposition \ref{prop:unaire} we show that for every cardinality relation $R(X)$ which is not definable in $(\N,<)$, the relation $\eqcard$ is definable in $\mbox{WMSO}(\N,<,R)$ by the formula 
$F(X,Y)$:
$$ \forall Z ( Z \cap (X \cup Y)= \emptyset \rightarrow (R(X \cup Z) \leftrightarrow R(Y \cup Z))).$$
This does not hold anymore for finite models, since now 
the sets $Z$ are limited to subsets of the (finite) domain, thus it can happen that two subsets $X,Y$ of the domain with different cardinality cannot be distinguished with such small sets $Z$. The issue is similar for the formula (\ref{eq:FF}) used in the proof of Theorem \ref{thm:MV}.

We shall prove undecidability of the satisfiability problem for $MSO_{\{<,R\}}$ over finite orderings by defining a sufficiently big fragment of the $\eqcard$ relation and then use Proposition \ref{prop:equi-undec-fini}.

\subsection{Definability of fragments of the equi-cardinality relation}

We consider formulas  which capture (arbitrarily great) fragments of $EqCard$.

\begin{defi}\label{def:quasi}
Let $n \geq 1$, $R \subseteq {\mathcal F}^n$ be a cardinality relation, and ${\mathcal L}=\{<,R\}$. 
Given a $\mathcal L$-formula $\varphi(X,Y)$, we say that $\varphi$ defines a {\em quasi-equicardinality} relation if the following properties hold:

\begin{enumerate}
\item $Val_q(\varphi)$ is a cardinality relation for every $q \geq 0$
\item $\I_q(\varphi) \subseteq \I_q(EqCard)$   for every $q \geq 0$
\item For every $k \geq 1$ there exists $Q \geq 1$ such that for every $q \geq Q$ we have 
$$\I_q(\varphi)  \cap ([0,k) \times [0,k))= \I_q(EqCard)  \cap ([0,k) \times [0,k))$$
(i.e. for $q$ great enough, the interpretations of $\varphi$ and $EqCard$ in ${\mathcal M}_q$ coincide for subsets of size less than $k$).

\end{enumerate}
\end{defi}

\begin{prop}\label{prop:quasieq}
Let $n \geq 1$ and let $R \subseteq {\mathcal F}^n$ be a cardinality relation which is not definable in $(\N,<)$. There exists a $\{<,R\}$-formula $\varphi$ which defines a quasi-equicardinality relation.
\end{prop}

\begin{proof}

We first prove the claim for every cardinality relation $R \subseteq {\mathcal F}^n$ such that $\I(R)$ is not STRUP.

We re-use the formula  $FF(X,Y)$ which was introduced in the proof of Theorem \ref{thm:MV}, and which was defined as 
{\setlength\arraycolsep{3pt}
\begin{eqnarray}
\lefteqn{
\forall Z_1 \dots \forall Z_n (\bigwedge_{i=1}^n (Z_i \cap (X \cup Y)= \emptyset \rightarrow
} 
\\
& & (R(Z_1,\dots, Z_{i-1}, Z_i \cup X,Z_{i+1},\dots,Z_n) 
\leftrightarrow  R(Z_1,\dots, Z_{i-1}, Z_i \cup Y,Z_{i+1},\dots,Z_n))). \nonumber
\end{eqnarray}
}%
We claim that the $\{<,R\}$-formula $\varphi(X,Y):$
$$FF(X,Y) \wedge \neg\exists X' (X' \subsetneq X \wedge FF(X',Y)) \wedge \neg\exists Y' (Y' \subsetneq Y \wedge FF(X,Y'))$$
defines a quasi-equicardinality relation. 

The fact that $\varphi$ defines a cardinality relation for every $q \geq 0$ follows easily from the fact that this is true for $FF(X,Y)$.

Let us now prove that $\I_q(\varphi) \subseteq \I_q(EqCard)$ 
for every $q \geq 0$. Assume that $\varphi(X,Y)$ holds for some sets $X,Y \subseteq [0,q)$. If $|X| \ne |Y|$, say $|X| < |Y|$ without loss of generality, then if we choose $Y' \subsetneq Y$ such that $|Y'|=|X|$ then $FF(X,Y')$ holds. This implies that $\exists Y' (Y' \subsetneq Y \wedge FF(X,Y'))$ holds, which contradicts the fact that $\varphi(X,Y)$ holds. 

It remains to prove that $\varphi$ satisfies item $(3)$ of Definition \ref{def:quasi}. Let $k \geq 1$. By our assumption $\I(R)$ is not STRUP, thus for every $p \geq 1$, if we set $\mu=(0,p,p,\dots,p)$ and $m=k$ in Definition \ref{def:strup}, then we know that there exists some $n$-tuple $\overline{x}_p=(x_{p,1},\dots,x_{p,n}) \in \N^n$ and some $j_p \in [1,n]$ such that $x_{p,i} \geq k$ for every $i$, and exactly one element among $\overline{x}_p$ and 
$(x_{p,1},\dots,x_{p,j_p-1},x_{p,j_p}+p,x_{p,j_p+1},\dots,x_{p,n})$ belongs to $\I(R)$. 

We set $Q=2k+\max\{ x_{p,i}  \  | 1 \leq p \leq k, i \in [1,n] \}$. 

Assume that $q \geq Q$, and let $X,Y \subseteq [0,q)$ be such that $|X|=|Y| \leq k$. We have to show that $\varphi(X,Y)$ holds in ${\mathcal M}_q$. It is clear that $FF(X,Y)$ holds. Let us prove that $\exists X' (X' \subsetneq X \wedge FF(X',Y))$ does not hold (the proof that $\exists Y' (Y' \subsetneq Y \wedge FF(X,Y'))$ does not hold is similar). Indeed assume that there exists a strict subset $X'$ of $X$ such that $FF(X',Y)$ holds. Let $p=|Y|-|X'|$. We claim that there exist sets $Z_1, \dots, Z_n \subseteq [0,q)$ such that 
\begin{itemize}
\item $|Z_i|=x_{p,i}$ for every $i \ne j_p$;
\item  $Z_{j_{p}}$ is disjoint from $X' \cup Y$ and satisfies $|Z_{j_{p}}|=x_{p,j_p}-|X'|$.
\end{itemize}
Indeed, for every $i \ne j_p$ the existence of $Z_i$ follows from the fact that $x_{p,i} \leq Q \leq q$. In order to prove that $Z_{j_p}$ exists it suffices to prove that 
$$|[0,q) \setminus (X' \cup Y)| \geq x_{p,j_p}-|X'| \geq 0.$$
On the one hand we have
$$x_{p,j_p}-|X'| \geq x_{p,j_p}-|X| \geq x_{p,j_p}-k \geq 0$$ 
and on the other hand we have
$$|[0,q) \setminus (X' \cup Y)| \geq q-2k \geq Q-2k \geq x_{p,j_p} \geq x_{p,j_p} - |X'|.$$

Now we have $|X' \cup Z_{j_{p}}|=x_{p,j_p}$ and $|Y \cup Z_{j_{p}}|=|Y|+|Z_{j_{p}}|=x_{p,j_p}+p$. By definition of $\overline{x}_p$, exactly one element among
$(x_{p,1},\dots,x_{p,n})$ and $(x_{p,1},\dots,x_{p,j_p-1},x_{p,j_p}+p,x_{p,j_p+1},\dots,x_{p,n})$ belongs to $\I(R)$, that is, exactly one element among $(Z_1, \dots,  Z_{j_p-1},Z_{j_p} \cup X', Z_{j_p+1}, \dots, Z_n)$ and $(Z_1, \dots,  Z_{j_p-1}, Z_{j_p} \cup Y, Z_{j_p+1}, \dots, Z_n)$ belongs to $R$, hence the formula 
$$R(Z_1,\dots, Z_{j_p-1},Z_{j_p} \cup Y,Z_{j_p+1},\dots,Z_n) \leftrightarrow R(Z_1,\dots, Z_{j_p-1}, Z_{j_p} \cup Y,Z_{j_p+1},\dots,Z_n)$$ 
does not hold, which contradicts our assumption  that $FF(X',Y)$ holds. 

We proved the claim for every cardinality relation $R \subseteq {\mathcal F}^n$ such that $\I(R)$ is not STRUP. Now we can prove the claim for every cardinality relation $R \subseteq {\mathcal F}^n$ by induction over $n$. The case $n=1$ follows from the above result  since if $R$ is not definable in $(\N,<)$ then $\I(R)$ is not STRUP. For the induction step we use Theorem \ref{thm:MV}:  either some section of $\I(R)$ is not definable in $(\N,<)$, or $\I(R)$ is not STRUP. For the  latter case we use again the above result, and for the former case the result follows from the induction hypothesis and the fact that every section of $\I(R)$ is Card-definable in $(\N,<,R)$.
\end{proof}

\subsection{Undecidability of the satisfiability problem}

We can state the main result of Section \ref{sec:fini}.

\begin{thm}\label{thm:indecfini}
Let $n \geq 1$, and let $R \subseteq {\mathcal F}^n$ be a cardinality relation which is not definable in $(\N,<)$. The satisfiability problem for $MSO_{\{<,R\}}$ over finite orderings is undecidable.  
\end{thm}

\begin{proof}
We proceed by reduction from the satisfiability problem for $MSO_{\{<,EqCard\}}$, which is undecidable by Proposition \ref{prop:equi-undec-fini}.
Let $R \subseteq {\mathcal F}^n$ be a cardinality relation which is not definable in $(\N,<)$. By Proposition \ref{prop:quasieq} there exists a quasi-equicardinality relation $\varphi$ which is $\{<,R\}$-definable. Consider the $\{<,R\}$-formula $GoodInitSeg(X)$ defined as 
$$\exists z  \forall x (x \in X \leftrightarrow x \leq z) \wedge \forall Y \subseteq X  \ \varphi(Y,Y).$$
This formula expresses in every structure ${\mathcal M}_q$ that $X$, on the one hand, is a nonempty initial segment of $[0,q)$, and on the other hand that $(y,y) \in {\mathcal I}_q(\varphi)$ for every $y \leq |X|$. The latter property, combined with the fact that 
${\mathcal I}_q(\varphi) \subseteq {\mathcal I}_q(EqCard)$, ensures that 
$${\mathcal I}_q(\varphi) \cap (X \times X) = {\mathcal I}_q(EqCard) \cap (X \times X),$$
i.e. that $\varphi$ coincides with $EqCard$ for all subsets of $X$. 

Now let $\theta$ be the function which maps every  $\{<,EqCard\}$-sentence $G$ to the $\{<,R\}$-sentence $\theta(G)$ defined as:
$$\exists X \ (GoodInitSeg(X) \wedge {G^*})$$
where ${G^*}$ is obtained from $G$ by relativizing all quantifiers to $X$ and by replacing every occurrence of the predicate $EqCard$ by $\varphi$. 

Let us prove that $G$ is satisfiable iff $\theta(G)$ is. Assume first that there exists $k\geq 1$ such that $G$ holds in ${\mathcal M}_k$. The formula $\varphi$ defines a quasi-equicardinality relation, thus there exists $q \geq k$ such that in ${\mathcal M}_q$, the interpretation of $\varphi$ coincides with $EqCard$ on the initial segment $X=[0,k)$. This implies that $\theta(G)$ holds in ${\mathcal M}_q$, thus $\theta(G)$ is satisfiable. Conversely if $\theta(G)$ holds in some structure 
${\mathcal M}_q$, for some initial segment $X$ of size $k\leq q$, then it follows from the very construction of $\theta(G)$ that $G$ holds in ${\mathcal M}_k$, thus $G$ is satisfiable.
\end{proof}

 The proof of Theorem \ref{thm:indecfini} relies on the possibility to find in ${\mathcal M}_q$ an (arbitrarily great) initial segment of $[0,q)$ where $EqCard$ is definable, whence where $+$ and $\times$ are definable (by Proposition \ref{prop:equi-undec-fini}). This raises the question whether $+$ and $\times$ are $\{<,R\}$-finite-definable. The answer seems to depend on the choice of $R$, but we were not able to prove a general result. However, 
we exhibit below two examples of unary relations $R$ for which the answer is positive, namely when ${\mathcal I}(R)$ stands for the set of powers of $2$, and for the set of prime numbers. The study of these examples was suggested by Bruno Courcelle. For each of them, the strong arithmetical properties of $\I(R)$ allow to ``control" the initial segment on which $EqCard$ is definable, which enables to prove definability of $+$ and $\times$ in a direct way.

\begin{prop}
Let $R(X)$ be interpreted as ``$|X|$ is a power of $2$". The relation $\eqcard$, as well as the graphs of $+$ and $\times$, are $\{<,R\}$-finite-definable. 

\end{prop}

\begin{proof}
By Proposition \ref{prop:equi-undec-fini} it suffices to prove that $\eqcard$ is $\{<,R\}$-definable. 
Without loss of generality we can limit ourselves to structures with size $q \geq 4$.

We shall re-use some $\{<,R\}$-formulas introduced earlier. Recall that $F(X,Y)$ is the formula
$$ \forall Z ( Z \cap (X \cup Y)= \emptyset \rightarrow (R(X \cup Z) \leftrightarrow R(Y \cup Z))).$$
and that the $\{<,R\}$-formula $\varphi(X,Y):$
$$F(X,Y) \wedge \neg\exists X' (X' \subsetneq X \wedge F(X',Y)) \wedge \neg\exists Y' (Y' \subsetneq Y \wedge F(X,Y'))$$
defines a quasi-equicardinality relation by Proposition \ref{prop:quasieq} (case $n=1$).

Let $q \geq 4$, and let $p$ be such that $2^{p+2} \leq q < 2^{p+3}.$
We shall prove first that in ${\mathcal M}_q$, the relation $\varphi(X,Y)$ coincides with $EqCard$ for subsets with size at most $2^p$. It is sufficient to prove that $\varphi(X,Y)$ does not hold if $|X|<|Y|\leq 2^p$. Let $Z \subseteq [0,q)$ be disjoint from $X \cup Y$ and such that $|Z|=2^{p+1}-|Y|$. Such a set exists since 
$$|[0,q) \setminus (X \cup Y)| \geq q-(|X|+|Y|) \geq q-2^{p+1} \geq 2^{p+1}.$$
On the one hand, we have $|Y \cup Z|=2^{p+1}$ thus $R(Y \cup Z)$ holds. On the other hand,
$$|X\cup Z|=|X|+2^{p+1}-|Y|=2^{p+1}+(|X|-|Y|)$$ which implies 
$$2^p<|X\cup Z|<2^{p+1}$$
thus  $R(X \cup Z)$ does not hold.
It follows that $F(X,Y)$ does not hold, as well as $\varphi(X,Y)$.

Now every subset of $[0,q)$ can be written as the disjoint union of $2^3$ subsets of size at most $2^p$. This enables to define $EqCard(X,Y)$, by a formula which expresses that $X$ and $Y$ can be written as $X= \dot{\cup}_{1 \leq i \leq 8} X_i$ and $Y= \dot{\cup}_{1 \leq i \leq 8} Y_i$ where $\varphi(X_i,Y_i)$ holds for every $i \in [1,8]$. 
\end{proof}

Let us consider the second example.

\begin{prop}
Let $R(X)$ be interpreted as ``$|X|$ is a prime number". The relation $\eqcard$, as well as the graphs of $+$ and $\times$, are $\{<,R\}$-finite-definable. 
\end{prop}

\begin{proof} 
By Proposition \ref{prop:equi-undec-fini} it suffices to prove that $\eqcard$ is $\{<,R\}$-definable.
Without loss of generality we can limit ourselves to structures with size $q \geq 6$.

We consider again the formula $\varphi(X,Y)$ used in the previous proof. Let  $q \geq 6$. We first prove that in ${\mathcal M}_q$, the interpretation of $\varphi(X,Y)$ and $EqCard(X,Y)$ coincide for subsets $X,Y$ such that $|X|$ and $|Y|$ are primes greater than $3$ and less than or equal to $\frac{q}{3}$.

It suffices to prove that for all subsets $X,Y$ of $[0,q)$ such that $|X|=p_1$ and $|Y|=p_2$ where $p_1,p_2$ are two primes such that $3<p_1 < p_2 \leq \frac{q}{3}$, the relation $\varphi(X,Y)$ does not hold. 

Let $Z$ be disjoint from $X \cup Y$ and such that $|Z|=p_2-p_1$. Such a set exists since $|[0,q) \setminus (X \cup Y)| \geq \frac{q}{3} $ by our assumption. 
On the one hand, we have $|X \cup Z|=p_2$, thus $R(X \cup Z)$ holds. On the other hand we have $ |Y \cup Z|=p_1+2(p_2-p_1)$. Now there is at least a multiple of 3 among the three integers $p_1$, $p_1+(p_2-p_1)$ and $p_1+2(p_2-p_1)$, and this cannot be the two first ones since they are primes greater than 3. It follows that $p_1+2(p_2-p_1)$ is a multiple of 3, thus is not prime since it is greater than 3. Therefore
$R(Y \cup Z)$ does not hold. Hence $F(X,Y)$ does not hold, as well as $\varphi(X,Y)$.

By \cite{Ramare95} every integer $k \geq 2$ can be written as a sum of at most $7$ primes. Thus we can define the restriction of $\eqcard$ to subsets of size less than or equal to $\frac{q}{3}$ by a formula $H(X,Y)$ which expresses that, either $|X|=|Y|=v$ with $v\leq 1$, or there exists $j\in [1,7]$ such that $X$ and $Y$ can be written as $X= \dot{\cup}_{1 \leq i \leq j} X_i$ and $Y= \dot{\cup}_{1 \leq i \leq j} Y_i$ where the formula $$R(X_i) \wedge R(Y_i) \wedge (\varphi(X_i,Y_i) \vee \bigvee_{2 \leq k \leq 3}(Card_k(X_i) \wedge Card_k(Y_i)))$$
 holds for every for every $i \in [1,j]$.

Now we have $q\geq 6$ thus $4 \lfloor \frac{q}{3} \rfloor \geq q$, hence each subset of $[0,q)$ can be written as the disjoint union of 4 subsets of size less than or equal to $\frac{q}{3}$. This allows to define $EqCard(X,Y)$ by a formula which expresses that $X$ and $Y$ can be written as $X= \dot{\cup}_{1 \leq i \leq 4} X_i$ and $Y= \dot{\cup}_{1 \leq i \leq 4} Y_i$ where $H(X_i,Y_i)$ holds for every $i \in [1,4]$.
\end{proof}

\section{Open problems}

As discussed in Remark \ref{rem:FO}, one can deduce from Theorem \ref{thm:indeci} that for every cardinality relation $R$ such that $cod(R)$ is not FO-definable in $(\N,+,V_2)$, the function $\times$ is FO-definable in $(\N,+,V_2,cod(R))$, which yields undecidability of the theory. We noticed that this holds for instance for the structure  $(\N,+,V_2,EqNonZeroBits)$ where $EqNonZeroBits(x,y)$ holds if $x$ and $y$ have the same number of non-zero bits. In this particular case, and more generally when $R$ is unary, we can actually prove that $(\N,+,cod(R))$ already suffices to FO-define $\times$ and get undecidability. We do not know whether this still holds for any cardinality relation $R$.  More generally it would be interesting to study the expressive power of fragments of FO arithmetic which include predicates like $EqNonZeroBits$.

Another series of questions is related to our Michaux-Villemaire-like Theorem \ref{thm:MV}. The result is stated for cardinality relations. Does it hold for any predicate $R$ ? More generally, which structures enjoy a similar property, and can we characterize them in model-theoretic terms ? If we consider FO logic,  besides $(\N,+)$, one can show easily that the property holds e.g. for $(\N,+,\times)$ and $(\N,+,x \mapsto 2^x)$, while it does not hold for $\mbox{FO}(\N,x \mapsto x+1)$ (see \cite[Section 2.6]{bes02a}). Recently Arthur Milchior \cite{Milchior13} proved a variant of Michaux-Villemaire theorem for $\mbox{FO}(\N,<,\{x \equiv y \pmod k\}_{k \geq 2})$
and use this to specify the decidability frontier for satisfiability of FO logics over words which lay between $\mbox{FO}_{\{<\}}$ and $\mbox{FO}_{\{+\}}$.

\section*{Acknowledgements.} 
We wish to thank
\begin{itemize} 
\item Bruno Courcelle for raising questions which led to this work, and for his help during the preparation of the paper;
\item Arthur Milchior for his careful reading of a preliminary version of the paper;
\item the organizers of the 30th Days of Weak Arithmetics held in 2011 at Paris 7 University, where we learned about Bruno Courcelle's question;   
\item the anonymous referee who suggested to use the formula (\ref{eq:FF}) in the proof of Theorem \ref{thm:MV} and Proposition \ref{prop:quasieq}, which led to a great simplification of the initial proofs.
\end{itemize}

\bibliographystyle{plain}
\bibliography{msocard}

\begin{thebibliography}{10}

\bibitem{DBLP:journals/jsyml/BaranyKR11}
V.~B{\'a}r{\'a}ny, L.~Kaiser, and A.~Rabinovich.
\newblock Expressing cardinality quantifiers in monadic second-order logic over
  chains.
\newblock {\em J. Symb. Log.}, 76(2):603--619, 2011.

\bibitem{Bes97a}
A.~B{\`e}s.
\newblock Undecidable extensions of {B}{\"u}chi arithmetic and
  {C}obham-{S}em{\"e}nov theorem.
\newblock {\em J. Symb. Log.}, 62(4):1280--1296, 1997.

\bibitem{bes02a}
A.~B\`es.
\newblock A survey of arithmetical definability.
\newblock In {\em {A tribute to Maurice Boffa}}, pages 1--54. Soc. Math.
  Belgique, 2002.

\bibitem{BHMV94}
V.~Bruy\`ere, G.~Hansel, C.~Michaux, and R.~Villemaire.
\newblock Logic and $p$-recognizable sets of integers.
\newblock {\em Bulletin of the Belgian Mathematical Society - Simon Stevin},
  1(2):191--238, 1994.

\bibitem{Buchi60}
J.~R. B{\"u}chi.
\newblock Weak second-order arithmetic and finite automata.
\newblock {\em Z. Math. Logik und grundl. Math.}, 6:66--92, 1960.

\bibitem{Buchi62}
J.~R. B{\"u}chi.
\newblock On a decision method in the restricted second-order arithmetic.
\newblock In {\em Proc. {I}nt. {C}ongress {L}ogic, {M}ethodology and
  {P}hilosophy of {S}cience, {B}erkeley 1960}, pages 1--11. Stanford University
  Press, 1962.

\bibitem{Colc13}
T.~Colcombet.
\newblock Regular cost functions, part i: Logic and algebra over words.
\newblock {\em Logical Methods in Computer Science}, 9(3), 2013.

\bibitem{Cour12}
B.~Courcelle and J.~Engelfriet.
\newblock {\em Graph structure and monadic second-order logic. A
  language-theoretic approach.}
\newblock Cambridge University Press, 2012.

\bibitem{Elgot61}
C.~C. Elgot.
\newblock Decision problems of finite automata design and related arithmetics.
\newblock {\em Trans. Amer. Math. Soc.}, 98:21--52, 1961.

\bibitem{ElgotRabin66}
C.~C. Elgot and M.~O. Rabin.
\newblock Decidability and undecidability of extensions of second (first) order
  theory of (generalized) successor.
\newblock {\em J. Symb. Log.}, 31(2):169--181, 1966.

\bibitem{FV59}
S.~Feferman and R.L. Vaught.
\newblock The first order properties of products of algebraic systems.
\newblock {\em Fundam. Math.}, 47:57--103, 1959.

\bibitem{Gurevich85}
Y.~Gurevich.
\newblock Monadic second-order theories.
\newblock In J.~Barwise and S.~Feferman, editors, {\em Model-Theoretic Logics},
  pages 479--506. Springer-Verlag, Perspectives in Mathematical Logic, 1985.

\bibitem{henriksen1995mona}
J.~Henriksen, J.~Jensen, M.~J{\o}rgensen, N.~Klarlund, R.~Paige, T.~Rauhe, and
  A.~Sandholm.
\newblock Mona: Monadic second-order logic in practice.
\newblock In E.~Brinksma, R.~Cleaveland, K.~Guldstrand Larsen, T.~Margaria, and
  B.~Steffen, editors, {\em TACAS}, volume 1019 of {\em Lecture Notes in
  Computer Science}, pages 89--110. Springer, 1995.

\bibitem{kaufmann1985quantifier}
M.~Kaufmann.
\newblock The quantifier "there exist uncountably many" and some of its
  relatives.
\newblock {\em Perspectives in Mathematical Logic. Model Theoretic Logics,
  chap. 4, Springer Verlag}, 1985.

\bibitem{KlaedtkeR03}
F.~Klaedtke and H.~Rue{\ss}.
\newblock Monadic second-order logics with cardinalities.
\newblock In J.~C.~M. Baeten, J.~K. Lenstra, J.~Parrow, and G.~J. Woeginger,
  editors, {\em ICALP}, volume 2719 of {\em Lecture Notes in Computer Science},
  pages 681--696. Springer, 2003.

\bibitem{KNR05}
V.~Kuncak, H.~H. Nguyen, and M.~Rinard.
\newblock An algorithm for deciding bapa: Boolean algebra with {P}resburger
  arithmetic.
\newblock In {\em Proc. CADE-20}, volume 3632 of {\em Lect. Notes in Comput.
  Sci.}, pages 260--277, 2005.

\bibitem{MV96}
C.~Michaux and R.~Villemaire.
\newblock Presburger arithmetic and recognizability of sets of natural numbers
  by automata: New proofs of {C}obham's and {S}em\"enov's theorems.
\newblock {\em Annals of Pure and Applied Logic}, 77(3):251--277, 1996.

\bibitem{Milchior13}
A.~Milchior.
\newblock Undecidability of satisfiability of expansions of {$FO[<]$} with a
  semilinear non regular predicate over words.
\newblock in preparation.

\bibitem{Muc03}
An.~A. Muchnik.
\newblock The definable criterion for definability in {P}resburger arithmetic
  and its applications.
\newblock {\em Theoretical Computer Science}, 290(3):1433--1444, 2003.

\bibitem{Rabinovich07a}
A.~Rabinovich.
\newblock On decidability of monadic logic of order over the naturals extended
  by monadic predicates.
\newblock {\em Inf. Comput}, 205(6):870--889, 2007.

\bibitem{Ramare95}
O.~Ramar{\'e}.
\newblock On {S}nirel'man's constant.
\newblock {\em Annali della Scuola Normale Superiore di Pisa-Classe di
  Scienze}, 22(4):645--706, 1995.

\bibitem{Rev05}
P.~Revesz.
\newblock The expressivity of constraint query languages with boolean algebra
  linear cardinality constraints.
\newblock In {\em Proc. ADBIS'04}, volume 3631 of {\em Lect. Notes in Comput.
  Sci.}, pages 167--182, 2005.

\bibitem{Robinson58}
R.M. Robinson.
\newblock {Restricted set-theoretical definitions in arithmetic.}
\newblock {\em Proc. Am. Math. Soc.}, 9:238--242, 1958.

\bibitem{Thomas97a}
W.~Thomas.
\newblock Languages, automata, and logic.
\newblock In G.~Rozenberg and A.~Salomaa, editors, {\em Handbook of Formal
  Languages}, volume III, pages 389--455. Springer-Verlag, 1997.

\bibitem{Trakh50}
B.~A. Trakhtenbrot.
\newblock Impossibility of an algorithm for the decision problem in finite
  classes.
\newblock {\em Doklady Akademii Nauk SSSR}, 70:569--572, 1950.
\newblock (in Russian).

\bibitem{Trakh61}
B.~A. Trakhtenbrot.
\newblock Finite automata and logic of monadic predicates (in {R}ussian).
\newblock {\em Dokl.\ Akad.\ Nauk SSSR}, 140:326--329, 1961.

\bibitem{villemaire1992joining}
R.~Villemaire.
\newblock Joining k-and l-recognizable sets of natural numbers.
\newblock In {\em Proceedings of the 9th Annual Symposium on Theoretical
  Aspects of Computer Science}, pages 83--94. Springer-Verlag, 1992.

\end{thebibliography}

\end{document}